\documentclass{article}
\usepackage{amsmath}
\usepackage{mathtools}
\usepackage{amssymb}
\usepackage{amsthm}
\usepackage{graphicx}
\usepackage{subcaption}
\usepackage{tikz}
\usetikzlibrary{arrows.meta, positioning, calc, shapes.geometric, shapes.arrows, fit, backgrounds}
\usepackage{booktabs}
\usepackage{natbib}

\usepackage{arxiv}
\renewcommand{\headeright}{Technical Report}
\renewcommand{\undertitle}{Technical Report}
\renewcommand{\shorttitle}{Predictive Coding via Proximal Gradients}

\usepackage{hyperref}
\hypersetup{colorlinks=true, linkcolor=gnblue4, breaklinks=true, urlcolor=gnblue4, citecolor=gnblue4}

\usepackage{xcolor}
\definecolor{fblightblue}{RGB}{220,235,250}

\newtheorem{proposition}{Proposition}
\newtheorem{lemma}[proposition]{Lemma}
\newtheorem{remark}[proposition]{Remark}

\newcommand{\until}[1]{\{1,\dots, #1\}}
\newcommand{\subscr}[2]{#1_{\textup{#2}}}

\newcommand{\prox}{\operatorname{prox}}

\newcommand{\real}{\ensuremath{\mathbb{R}}}

\newcommand\oprocendsymbol{\hbox{$\square$}}
\newcommand\oprocend{\relax\ifmmode\else\unskip\hfill\fi\oprocendsymbol}

\DeclareSymbolFont{bbold}{U}{bbold}{m}{n}
\DeclareSymbolFontAlphabet{\mathbbold}{bbold}
\newcommand{\vect}[1]{\mathbbold{#1}}
\newcommand{\vectorones}[1][]{\vect{1}_{#1}}

\newcommand{\argmin}{\ensuremath{\operatorname{argmin}}}
\newcommand{\argmax}{\ensuremath{\operatorname{argmax}}}

\newcommand{\relu}{\operatorname{relu}}
\newcommand{\soft}{\operatorname{soft}}
\newcommand{\sat}{\operatorname{sat}}

\newcommand{\E}{\mathbb{E}}
\newcommand{\e}{\mathrm{e}}

\definecolor{gnred}{RGB}{255,91,89}

\definecolor{gnred1}{RGB}{71,0,0}
\definecolor{gnred2}{RGB}{117,0,0}
\definecolor{gnred3}{RGB}{164,0,0}
\definecolor{gnred4}{RGB}{211,0,0}
\definecolor{gnred5}{RGB}{255,0,0}
\definecolor{gnred6}{RGB}{255,42,34}
\definecolor{gnred7}{RGB}{255,91,89}

\definecolor{gnblue1}{RGB}{0,36,71}
\definecolor{gnblue2}{RGB}{0,60,118}
\definecolor{gnblue3}{RGB}{0,85,164}
\definecolor{gnblue4}{RGB}{0,108,212}
\definecolor{gnblue5}{RGB}{0,133,255}
\definecolor{gnblue6}{RGB}{35,156,255}
\definecolor{gnblue7}{RGB}{88,177,255}

\definecolor{gnbrown1}{RGB}{71,27,0}
\definecolor{gnbrown2}{RGB}{117,45,0}
\definecolor{gnbrown3}{RGB}{164,62,0}
\definecolor{gnbrown4}{RGB}{211,80,0}
\definecolor{gnbrown5}{RGB}{255,97,0}
\definecolor{gnbrown6}{RGB}{255,127,26}
\definecolor{gnbrown7}{RGB}{255,155,86}

\newcommand{\mcE}{\mathcal{E}}
\newcommand{\mcR}{\mathcal{R}}

\newcommand{\mcN}{\mathcal{N}}

\newcommand{\softmax}[1]{\operatorname{softmax}_{#1}}
\newcommand{\entropy}{\mathsf{H}}
\newcommand{\entropybar}{\subscr{\entropy}{barrier}}
\newcommand{\simplex}[1]{\Delta_{#1}}

\title{Predictive Coding with Bayesian Priors \\ via Proximal Gradients}

\author{\large Francesco Bullo\thanks{Francesco Bullo, Mechanical
    Engineering, 2325 Engineering Bldg II, UCSB, Santa Barbara, CA
    93106-5070, USA. Website \url{https://fbullo.github.io}. Email:
    \url{mailto:bullo@ucsb.edu}.} \\ Department of Mechanical
  Engineering\\ and Dynamical Neuroscience Program\\ UC Santa Barbara}

\begin{document}

\maketitle

\begin{abstract}
  We recast predictive coding as continuous-time proximal gradient descent
  applied to a regularized maximum-a-posteriori (MAP) objective.  We study
  first a single-level problem and then a multi-level hierarchy.  For the
  single-level problem, we show that proximal gradient descent is precisely
  a leaky firing-rate network: the membrane leak, the effective recurrent
  matrix, the local synaptic drive, and the static nonlinearity all follow
  from one optimization principle, and the resulting circuit is the one
  proposed by Rao and Ballard.  The prior selects the nonlinearity through
  its proximal operator, and the likelihood precision sets the gain on the
  observation.  For the hierarchy, we show that a classical
  variable-splitting relaxation of the deep MAP problem yields hierarchical
  predictive coding as the interconnection of local and distributed
  solvers.  In probabilistic modeling terms, this relaxation replaces the
  directed generative chain by an undirected Markov random field whose node
  potentials are the level-wise priors.  Each level then applies its own
  activation function, namely the proximal operator of its prior.
\end{abstract}



\section{Introduction}\label{sec:intro}

\paragraph{Context and Problem Description.}\label{sec:context}
The view of the brain as a probabilistic inference machine originates with
Helmholtz's notion of perception as unconscious inference \citep{HvH:1867}.
Modern predictive coding turns this view into a concrete proposal: neural
circuits implement a generative model of the sensory data and reduce the
discrepancy between top-down predictions and bottom-up observations
\citep{RPNR-DHB:99,KJF:05,KJF:10}.  The Free Energy Principle extends the
same picture and asserts that any self-organizing biological system
minimizes variational free energy \citep{KJF:10,CLB-CSK-SM-AKS:17}.  In
parallel with these developments, computational neuroscience has identified
a small number of recurring operations, such as linear filtering,
thresholding, normalization~\citep{MC-DJH:12}, and sparse
coding~\citep{BAO-DJF:96}.  The aim of this paper is twofold.  First, we
clarify the relationship among (i) the Bayesian priors that the brain
assumes, (ii) the dynamical system that its neurons implement, and (iii)
the static activation function that those neurons apply.  Second, we use
this relationship to derive a hierarchical predictive coding architecture
from a single regularized maximum-a-posteriori (MAP) objective.

\paragraph{Literature Review.}\label{sec:literature}
\emph{Predictive coding and the Bayesian brain.}  In their groundbreaking
work \cite{RPNR-DHB:99} introduced predictive coding for the visual cortex.
These ideas were cast within a hierarchical Bayesian framework and the Free
Energy Principle by \cite{KJF:05,KJF:08,KJF:10}.  The variational
interpretation of predictive coding is reviewed in
\citep{CLB-CSK-SM-AKS:17,BM-AS-CLB:21,TS-AM-CLB-TL-RPNR-KF-AO:26}.  The
canonical circuit account of \citep{AMB-WMU-RAA-GRM-PF-KJF:12} (see
also~\citep{GBK-TDM:18}) associates the prediction-error and state
variables with distinct neural populations, with experimental support
reported in~\citep{AA-BW-GBK:17}.

\emph{Level-specific priors.}  \cite{TSL-DM:03} formalized the cortex as
hierarchical empirical Bayes with level-specific priors over edges, parts,
and objects.  \cite{KJF:08} made each level an empirical prior on the level
below.  The mixed models in~\citep{KJF-TF-FR-PS-GP:17,KJF-TP-BdV:17} combine
continuous predictive coding at lower levels with discrete (categorical)
Markov decision processes at higher levels, and they use two distinct
inference methods.  We refer the reader to~\citep{LFB-EKM:26} for a recent
perspective on categorization.  In contrast, the present paper recovers both
regimes from a single hierarchical proximal-gradient dynamics in which the
only level-dependent ingredient is the prior.

\emph{Predictive coding and variational autoencoders.}  A comprehensive
comparative study of Rao--Ballard predictive coding and the modern
variational autoencoder (VAE) was developed by \cite{JM:22}, who identifies
iterative versus amortized inference as the main distinction between the
two.  In the framework of this paper, a sparse-coding VAE (Laplace prior),
a nonnegative autoencoder (nonnegativity prior), and a softmax-classifier
head (shifted-entropy-barrier prior) are all instances of the same
hierarchical proximal-gradient dynamics. As described
by~\citep{JCRW-RB:17,BM-AT-CLB:22}, predictive coding with local Hebbian
plasticity approximates backpropagation along arbitrary computation graphs,
including autoencoder topologies.

\emph{Variable splitting.}  The auxiliary-variables view of a deep model
dates to the ``method of auxiliary coordinates'' by~\cite{MCP-WW:14} and
was developed for ADMM-based deep-network training by
\cite{GT-AB-MX-MS-AP-TG:16}. Other notable references on variable
splitting, ADMM, and proximal algorithms include
\citep{SP-NP-EC-NP-JE:10,NP-SB:14,PLC-JCP:11}.

\emph{Sparse coding and proximal-gradient algorithms.}  \cite{BAO-DJF:96}
showed that sparse-coding objectives applied to natural images learn
filters that resemble those found in early visual cortex.
\cite{CJR-DHJ-RGB-BAO:08} introduced the locally-competitive algorithm
(LCA), a recurrent network in membrane-potential form that solves the
sparse-coding problem.  Continuous-time proximal-gradient dynamics have
been recently introduced and analyzed in
\citep{SHM-MRJ:21,VC-AG-AD-GR-FB:23a,AG-AD-FB:24d,AD-VC-AG-GR-FB:23f,VC-AG-AD-GR-FB:23c};
in particular, \cite{VC-AG-AD-GR-FB:23a} propose a positive competitive
network similar to the LCA algorithm, but for firing rate networks.

\paragraph{Contributions and Discussion.}\label{sec:contributions}

After collecting preliminaries in Section~\ref{sec:prelim}, we revisit and
organize recent results that connect proximal gradient descent and
activation functions with predictive coding, and we then present our main
contribution on the hierarchical case.

We begin in Section~\ref{sec:firing-rate} by passing from unconstrained to
constrained inference (Proposition~\ref{prop:firing-rate}).  This
proposition is, for the most part, contained in the recent
work~\citep{VC-AG-AD-GR-FB:23a}; we restate it here, in a general
precision-weighted form, because it is the single-level building block for
the hierarchical construction that follows.  Standard predictive
coding~\citep{RPNR-DHB:99} performs unconstrained gradient descent on a
Gaussian data-fit energy.  We replace this gradient descent by a proximal
gradient dynamics, which accommodates nonsmooth and structured priors.
This single change yields an exact firing-rate network with (i) a passive
membrane leak, (ii) an effective recurrent matrix obtained from the
synaptic Gram matrix, (iii) a feedforward drive from the precision-weighted
observation, and (iv) a static nonlinearity equal to the proximal operator
of the prior. 

In our treatment of predictive coding, the static nonlinearities of
firing-rate neurons need not be postulated; they arise as the proximal
operators of specific Bayesian priors.  The correspondence between a prior
and its proximal operator (Laplace to soft-threshold, nonnegativity to
ReLU) is classical in the proximal-splitting literature
\citep{PLC-JCP:11,NP-SB:14}, and even the shifted-entropy-barrier to
softmax entry is known: it appears as \citep[Exercise~2.23]{PLC-JCP:20b}
and is developed by \citet{FR-VC-FB-GR:25k}.  Our contribution is not this
dictionary itself, but its interpretation as the catalog of activation
functions of a single predictive-coding network based on proximal
gradients.  The full catalog appears in Subsection~\ref{sec:prelim-priors}.

Our main contribution is the interpretation of hierarchical predictive
coding as distributed proximal gradient descent
(Proposition~\ref{prop:hier-pgd} in Section~\ref{sec:hier}).  We show that
the centralized MAP problem for a deep linear generative model admits, by
means of a classical variable-splitting relaxation, a reformulation as the
joint energy of hierarchical predictive coding.  Applied to a regularized
MAP objective, the splitting produces precisely the Rao--Ballard
prediction-error circuit, in which the level-wise prior plays the role of
the splitting penalty and selects the activation function at that level
through its proximal operator.  The single-level construction extends to a
deep hierarchy by means of three facts.  First, proximal gradient descent
on the joint hierarchical MAP energy is not a single combined update; it
decomposes exactly into a collection of local proximal-gradient steps, one
per level.  Second, each level runs in closed loop with its neighbors
alone: its only drive is the precision-weighted prediction error that it
exchanges with the level below and the level above, so that no global
coordinator is required.  Third, the levels are heterogeneous.  Each level
carries its own proximal operator $\prox_{\mcR^{(\ell)}}$ (and hence its
own activation function), its own time constant $\tau^{(\ell)}$, and its
own likelihood precision $\Pi^{(\ell)}$, so that different areas implement
different nonlinearities and gains within a single dynamics.  We envision
that our prior and activation function based generalization of Gaussian
predictive coding holds the promise to lead to more expressive and more
powerful models and improve the comparison of predictive coding performance
with that of autoencoders.

We do not review the convergence properties of proximal gradient descent
dynamics in this document; the existing analyses are summarized in
Remark~\ref{remark:convergence} below.

\paragraph{Paper Organization.}\label{sec:organization}
The rest of the paper is organized as follows.  Section~\ref{sec:prelim}
collects preliminaries on firing-rate networks, MAP estimation,
proximal-gradient preliminaries and the catalog of priors, proximal
operators, and activation functions.  Section~\ref{sec:firing-rate} derives
the single-level leaky firing-rate network
(Proposition~\ref{prop:firing-rate}) and presents three worked examples.
Section~\ref{sec:hier} extends the construction to a hierarchy by means of
variable splitting (Proposition~\ref{prop:hier-pgd}).

%% \clearpage
\section{Preliminaries}\label{sec:prelim}

This section collects the background used throughout the paper: the
firing-rate network model (§\ref{sec:prelim-fr}), maximum a-posteriori
inference (§\ref{sec:prelim-map}), proximal-gradient dynamics
(§\ref{sec:prelim-pgd}), and a catalog of canonical priors of random
variables, their proximal operators, and the firing-rate activation
functions they yield (§\ref{sec:prelim-priors}).

\subsection{Firing-Rate Networks}\label{sec:prelim-fr}
A \emph{firing-rate network} models a population of $n$ neurons whose state
$x \in \real^n$ evolves under a leak, recurrent interactions, and a
feedforward stimulus.  Given a \emph{stimulus} $u$, a \emph{synaptic matrix}
$W$, an \emph{input matrix} $B$, and an \emph{activation function} $\sigma$
applied component-wise, the dynamics are
\begin{equation}
  \dot{x} \;=\; \subscr{\mathsf{F}}{FR}(x,u) \;:=\; -x \;+\; \sigma\bigl(W x + B u\bigr) ,
  \label{eq:firing-rate}
\end{equation}
where $-x$ is the membrane leak, $W x$ the recurrent drive, $B u$ the
feedforward stimulus, and $\sigma$ the static nonlinearity.  A central aim of
this paper is to show that such networks are proximal gradient descent on a
regularized energy, so that the activation $\sigma$, the synaptic matrix $W$,
and the input matrix $B$ all follow from an optimization objective.

\subsection{Maximum A-Posteriori Inference}\label{sec:prelim-map}
Let $y$ be a measurement of a latent random variable $x$.  The
\emph{maximum a-posteriori} (MAP) estimate is, by definition, the mode of
the posterior $p(x \mid y)$:
\begin{equation}
  \subscr{x}{MAP} \;=\; \argmax_x \; p(x \mid y)
  \;=\; \argmin_x \; \{ -\log p(x \mid y) \} .
\end{equation}
By Bayes' theorem, $p(x \mid y) = p(y \mid x)\,p(x) / p(y)$, so that
$-\log p(x \mid y) = -\log p(y \mid x) - \log p(x) + \log p(y)$.  Since the
marginal $\log p(y)$ does not depend on $x$, it can be dropped from the
minimization, which yields
\begin{equation}
  \subscr{x}{MAP} \;=\; \argmin_x \; \{ -\log p(y \mid x) - \log p(x) \} ,
  \label{eq:map-general}
\end{equation}
where the first term is the negative log-likelihood (the ``data-fitting''
term) and the second term is the negative log-prior, which plays the role of
the regularizer.

For the linear-Gaussian model, the MAP problem~\eqref{eq:map-general}
reduces to a regularized least-squares problem, as follows.

\begin{lemma}[Linear-Gaussian MAP energy]\label{lem:gauss-map}
Let $y \in \real^N$ be a fixed (clamped) measurement and let the likelihood
be linear-Gaussian, $y \mid x \sim \mcN(\Phi x, \Pi^{-1})$, with dictionary
matrix $\Phi \in \real^{N \times M}$ and symmetric positive-definite
precision $\Pi \in \real^{N \times N}$; let the prior be $p(x) \propto
\exp(-\mcR(x))$ with $\mcR$ proper, closed, and convex.  The
\emph{prediction} is $\hat y(x) := \E[y \mid x] = \Phi x$ and the
\emph{prediction error} is $\epsilon(x) := y - \hat y(x)$.  Then the MAP
problem~\eqref{eq:map-general} reduces to minimizing the \emph{MAP energy}
\begin{equation}
  \subscr{\mcE}{MAP}(x) := \underbrace{\tfrac{1}{2}\|y - \hat y(x)\|_\Pi^2}_{\subscr{\mcE}{recon}(x)} + \mcR(x) ,
  \qquad \|v\|_\Pi^2 := v^\top \Pi\, v .
  \label{eq:map-prelim}
\end{equation}
The \emph{reconstruction energy} $\subscr{\mcE}{recon}(x) := \tfrac{1}{2}\|y
- \hat y(x)\|_\Pi^2 = -\log p(y \mid x) + \textnormal{const}$ is convex and
smooth, with gradient
\begin{equation}
  \nabla \subscr{\mcE}{recon}(x) = -\Phi^\top \Pi\, \epsilon(x) ,
  \label{eq:gauss-grad}
\end{equation}
Hessian $\nabla^2 \subscr{\mcE}{recon}(x) = \Phi^\top \Pi\, \Phi \succeq 0$,
and Lipschitz gradient with constant $L = \lambda_{\max}(\Phi^\top \Pi\,
\Phi)$.
\end{lemma}

\begin{proof}
  With $\hat y(x) = \Phi x$, the Gaussian density gives $-\log p(y \mid x)
  = \tfrac{1}{2}\|y - \hat y(x)\|_\Pi^2 +
  \tfrac{1}{2}\log\det(2\pi\Pi^{-1})$, whose second term does not depend on
  $x$, and the prior gives $-\log p(x) = \mcR(x) + \textnormal{const}$.
  Substituting both into~\eqref{eq:map-general} and discarding the
  $x$-independent constants yields the minimization of the MAP energy
  $\subscr{\mcE}{MAP}$ in~\eqref{eq:map-prelim}.  Differentiating
  $\subscr{\mcE}{recon}(x) = \tfrac{1}{2}(y - \Phi x)^\top \Pi (y - \Phi
  x)$ gives $\nabla\subscr{\mcE}{recon}(x) = -\Phi^\top \Pi (y - \Phi x) =
  -\Phi^\top \Pi\, \epsilon(x)$ and $\nabla^2\subscr{\mcE}{recon}(x) =
  \Phi^\top \Pi\, \Phi$. This quantity is positive semidefinite because
  $\Pi \succ 0$; the largest eigenvalue of this Hessian is the Lipschitz
  constant of the gradient.
\end{proof}

\subsection{Proximal Operators and Continuous-Time Proximal Gradient Descent}
\label{sec:prelim-pgd}
Consider a \emph{regularized optimization problem}
\begin{equation}
 \min_{x \in \real^M} \; \mcE(x) + \mcR(x),
\end{equation}
where the \emph{nominal energy} $\mcE$ is differentiable with
Lipschitz-continuous gradient, and the \emph{regularizer} $\mcR$ is proper,
closed, and convex (and possibly nonsmooth).  The \emph{proximal operator} of $\mcR$ is
\begin{equation}
    \prox_{\mcR}(v) \;:=\; \argmin_{z \in
      \real^M} \;\Bigl\{\, \mcR(z) + \tfrac{1}{2}\|z - v\|_2^2
    \,\Bigr\} ,
    \label{eq:prox-def}
\end{equation}
which is well-defined and single-valued under the stated assumptions
\citep{NP-SB:14, PLC-JCP:11}.  The continuous-time \emph{proximal gradient
descent} dynamics are
\begin{equation}
    \tau\, \dot x \;=\; -x \;+\; \prox_{\mcR}\bigl(x - \nabla \mcE(x)\bigr).
    \label{eq:pgd-generic}
\end{equation}
\emph{Intuition.}  The vector field in the right-hand side of
\eqref{eq:pgd-generic} pulls the state $x$ toward the $\mcR$-proximal point
obtained from a single step of gradient descent on $\mcE$: starting from
the current state $x$, one takes a gradient step $x \mapsto x -
\nabla\mcE(x)$ on the smooth part and then applies the proximal
operator of $\mcR$.

\subsection{Canonical Priors, Proximal Operators, and Activation Functions}
\label{sec:prelim-priors}

The proximal operators of canonical priors coincide, after rescaling, with
the static input-output nonlinearities used in firing-rate network models.
This correspondence has been recognized in special cases throughout the
sparse-coding and proximal-algorithms literature \citep{BAO-DJF:96,
  CJR-DHJ-RGB-BAO:08, NP-SB:14, VC-AG-AD-GR-FB:23a}.  We catalog these
canonical entries here for reference.

\paragraph{Activation functions.}
We use the following maps, with $\relu$, $\sat_{[0,1]}$, and
$\soft_\lambda$ applied entrywise to a vector argument: the \emph{rectified
linear unit} $\relu(v) := \max(0,v)$; the \emph{saturation}
$\sat_{[0,1]}(v) := \min\bigl(1,\max(0,v)\bigr)$; the \emph{soft-threshold}
$\soft_\lambda(v) := \mathrm{sign}(v)\,\max(0,|v|-\lambda)$, at threshold
$\lambda \geq 0$; and the \emph{softmax} $[\softmax{\theta}(v)]_i :=
\e^{v_i/\theta}/\bigl(\sum_j \e^{v_j/\theta}\bigr)$, at temperature $\theta
> 0$.  For a distribution $p \in \simplex{n}$, the \emph{entropy} is
$\entropy(p) := -\sum_i p_i \log p_i$.

\paragraph{Canonical correspondences.}
Table~\ref{tab:priors} summarizes the canonical triples of prior, proximal
operator, and activation function.

\begin{table}[htbp]
    \centering
    \small
    \setlength{\tabcolsep}{8pt}
    \begin{tabular}{lll}
        \toprule
        Prior/regularizer $\mcR(x)$ & $\prox_{\mcR}(v)$ & Activation \\
        \midrule
        $\tfrac{\lambda}{2}\,\|x\|_2^2$ (Gaussian) & $v/(1 + \lambda)$ & linear shrinkage \\
        $\lambda\, \|x\|_1$ (Laplace) & $\soft_\lambda(v)$ & soft-threshold  \\
        $\iota_{\{x \geq 0\}}(x)$ (nonnegativity) & $\relu(v)$ & ReLU \\
        $\iota_{[0,1]}(x)$ (box constraint) & $\sat_{[0,1]}(v)$ & clipped linear / saturation \\
        $\entropybar(x) = -\theta\,\entropy(x) - \tfrac{1}{2}\|x\|^2 + \iota_{\simplex{}}(x)$ & $\softmax{\theta}(v)$ & softmax  \\
        \bottomrule \smallskip
    \end{tabular}
    \caption{Correspondence between the prior/regularizer $\mcR$, the
      proximal operator $\prox_{\mcR}$, and the
      resulting static activation function of the proximal-gradient
      firing-rate network of §\ref{sec:firing-rate}.  The entropic row uses the
      \emph{shifted entropy barrier} with temperature $\theta \geq 1$:
      $\entropybar(x) := -\theta\,\entropy(x) -
      \tfrac{1}{2}\|x\|^2$, $x \in \simplex{n}$ (and $+\infty$
      otherwise); this function is closed, convex, and proper for $\theta\geq 1$,
      and its proximal operator is $\softmax{\theta}$ exactly
      \citep{FR-VC-FB-GR:25k}; see also \citet[Exercise~2.23]{PLC-JCP:20b}.}
    \label{tab:priors}
\end{table}

\paragraph{Gaussian prior (linear shrinkage).}
The Gaussian regularizer, $\mcR(x) = \tfrac{\lambda}{2}\|x\|_2^2$,
represents a basic energy constraint on the system.  Its proximal
operator is
$\prox_{\mcR}(v) = v/(1 + \lambda)$, which
acts as a linear shrinkage function and corresponds to neurons
operating in a linear, graded response regime.

\paragraph{Laplace prior (soft-thresholding).}
The Laplace prior, $\mcR(x) = \lambda\|x\|_1$, is the standard example for
sparse coding \citep{BAO-DJF:96}.  Its proximal operator is the
soft-thresholding function, which produces the LCA (membrane-potential)
dynamics of \citep{CJR-DHJ-RGB-BAO:08}.  This prior models the localized
filters found in early visual cortex.

\paragraph{Domain constraints (ReLU and saturation).}
Hard constraints represent physical limits on neural firing.  The
nonnegativity prior, $\mcR(x) = \iota_{\{x \geq 0\}}(x)$, yields the
standard ReLU activation function $\relu(v)$, which reflects that firing
rates cannot be negative.  When this is combined with an upper limit, the
box constraint $\mcR(x) = \iota_{[0,1]}(x)$ produces the clipped linear
threshold model $\sat_{[0,1]}(v)$.

\paragraph{Shifted entropy barrier (softmax).}
For categorical population codes, the shifted entropy barrier $\entropybar$
acts as a regularizer over the probability simplex \citep{FR-VC-FB-GR:25k}.
Its proximal operator is the softmax function, which biases the population
activity toward categorical representations.

%% \clearpage
\section{Predictive Coding as a Leaky Firing-Rate Network}
\label{sec:firing-rate}

We instantiate the maximum a-posteriori inference of
Subsection~\ref{sec:prelim-map} for the linear-Gaussian generative model in
Lemma~\ref{lem:gauss-map}, see \citep{RPNR-DHB:99, KJF:05}.  Let $y \in
\real^N$ be a clamped sensory observation and let $x \in \real^M$ be an
internal state (the ``neural activity'').  A \emph{dictionary matrix} $\Phi
\in \real^{N \times M}$ maps states to predictions via
\[
  y \;=\; \Phi x + \eta, \qquad \eta \sim \mathcal{N}(0,\,\Pi^{-1}),
\]
where $\Pi \in \real^{N \times N}$ is a symmetric positive-definite
\emph{likelihood precision}.  Recall the precision-weighted squared norm
$\|v\|_\Pi^2 = v^\top\Pi v$.  The derived quantities are the
\emph{prediction} $\hat{y} := \Phi x$ and the \emph{prediction error}
$\epsilon := y - \Phi x \in \real^N$.  The case $M > N$ is the
\emph{overcomplete} regime, as in the sparse-coding model
of~\cite{BAO-DJF:96}.

\paragraph{Standard predictive-coding dynamics and the MAP correction.}
Under a flat (uniform) prior on $x$, gradient descent on the negative
log-likelihood gives the standard predictive-coding dynamics
\citep{RPNR-DHB:99}
\begin{equation}
  \dot{x} \;=\; \Phi^\top \Pi\,\epsilon ,
  \label{eq:standard-pc}
\end{equation}
which lack a rest state and admit dense representations.  Both shortcomings
are removed by restoring the prior $\mcR$ and performing proximal-gradient
descent on the MAP energy $\subscr{\mcE}{MAP}$ of
Subsection~\ref{sec:prelim-map}.

\subsection{The Single-Level Firing-Rate Network}
\begin{proposition}[Predictive coding as a leaky firing-rate network]\label{prop:firing-rate}
Fix the following data:
\begin{itemize}
\item a dictionary matrix $\Phi \in \real^{N \times M}$ and a fixed observation $y \in \real^N$, for integers $N, M \geq 1$;
\item a likelihood precision $\Pi \in \real^{N \times N}$, symmetric and positive-definite ($\Pi \succ 0$);
\item a regularizer / negative log-prior $\mcR : \real^M \to \real \cup \{+\infty\}$, proper, closed, and convex;
\item a time constant $\tau > 0$.
\end{itemize}
Define the \emph{MAP energy}
\begin{equation}
  \subscr{\mcE}{MAP}(x) \;:=\;
  \subscr{\mcE}{recon}(x)+ \mcR(x) \;=\;
  \tfrac{1}{2}\|y - \Phi x\|_\Pi^2 +
  \mcR(x) , \qquad x \in \real^M ,
\end{equation}
and the \emph{prediction error}
$\epsilon(x) := y - \Phi x \in \real^N$.  Then continuous-time
proximal-gradient descent on $\subscr{\mcE}{MAP}$ is the firing-rate
network
\begin{equation}
\tau\,\dot{x} \;=\; -x \;+\;
\prox_{\mcR}\bigl(x + \Phi^\top \Pi\, \epsilon(x)\bigr) ,
\label{eq:pgd-network}
\end{equation}
or equivalently, in recurrent-network form,
\begin{equation}
    \tau\,\dot{x} \;=\; -x \;+\;
\prox_{\mcR}\bigl(W\, x + \Phi^\top \Pi\, y\bigr) ,
    \label{eq:pgd-recurrent}
\end{equation}
with leak $-x$, \emph{effective synaptic matrix}
$W := I_M - \Phi^\top \Pi\, \Phi \in \real^{M \times M}$,
feedforward drive $\Phi^\top \Pi\, y$, and static nonlinearity
$\prox_{\mcR}$.
\end{proposition}

\begin{proof}
\textit{Step 1: Gradient computation.}
The reconstruction energy is
$\subscr{\mcE}{recon}(x) = \tfrac{1}{2}\|y - \Phi x\|_\Pi^2 = \tfrac{1}{2}(y-\Phi x)^\top\Pi(y-\Phi x)$.
Differentiating with respect to $x$:
$$ \nabla_x \subscr{\mcE}{recon}(x) \;=\; -\Phi^\top \Pi\,(y - \Phi x) \;=\; -\Phi^\top \Pi\,\epsilon, $$
where $\epsilon := y - \Phi x$ is the prediction error.

\textit{Step 2: Substitution into the generic proximal gradient descent update.}
Applying \eqref{eq:pgd-generic} with the smooth part $\mcE$ instantiated as $\subscr{\mcE}{recon}$ gives
$\tau \dot{x} = -x + \prox_{\mcR}(x - \nabla\subscr{\mcE}{recon}(x))$.
Inserting $-\nabla\subscr{\mcE}{recon}(x) = \Phi^\top\Pi\epsilon$:
$$ \tau\dot{x} \;=\; -x + \prox_{\mcR}\bigl(x + \Phi^\top\Pi\,\epsilon\bigr), $$
which is \eqref{eq:pgd-network}.

\textit{Step 3: Rewriting the proximal argument.}
Expand the argument of the proximal operator by substituting $\epsilon = y - \Phi x$:
\begin{align*}
x + \Phi^\top \Pi\, \epsilon
&\;=\; x + \Phi^\top \Pi (y - \Phi x) \\
&\;=\; (I_M - \Phi^\top \Pi \Phi)\,x \;+\; \Phi^\top \Pi\, y \\
&\;=:\; W\,x + \Phi^\top\Pi\,y,
\end{align*}
with $W := I_M - \Phi^\top \Pi \Phi$.  Substituting back gives
\eqref{eq:pgd-recurrent}.
\end{proof}

\begin{remark}[Convergence]\label{remark:convergence}
  Convergence properties of the proximal-gradient dynamics
  \eqref{eq:pgd-network}--\eqref{eq:pgd-recurrent} have been studied
  using several approaches.  Contraction-theoretic results in the
  firing-rate case appear in \citep{VC-AG-AD-GR-FB:23c}; contractivity for
  general proximal-gradient dynamics with strongly convex $\mcE$ is treated
  in \citep{AD-VC-AG-GR-FB:23f}; linear exponential convergence estimates
  appear in \citep{VC-AG-AD-GR-FB:23a,VC-AD-AG-GR-FB:24a}; and monotone
  decrease, forward invariance, and a proximal Polyak--{\L}ojasiewicz
  analysis of LASSO-type problems are given in \citep{AG-AD-FB:24d}.  An
  integral-quadratic-constraint approach to global exponential stability of
  proximal-gradient and Douglas--Rachford splitting flows is given
  by~\cite{SHM-MRJ:21}.  We do not revisit these results here.
\end{remark}

\subsection{Examples and Hopfield Connection}
\label{subsec:examples}

We illustrate Proposition~\ref{prop:firing-rate} with three examples.  The
first example revisits the positive competitive network
of~\citep{VC-AG-AD-GR-FB:23a} from a Laplace-plus-nonnegativity prior.  The
second example revisits the softmax gradient
descent~\citep{FR-VC-FB-GR:25k} generated by a shifted-entropy-barrier
prior and performs categorical inference.  The third example shows that the
classical Hopfield energy, after an appropriate transformation, can be
regarded as an instance of the generic regularized energy of
Subsection~\ref{sec:prelim-pgd}.

\paragraph{Sparse coding and the positive competitive network.}
Consider an overcomplete dictionary $\Phi \in \real^{N\times M}$ ($M>N$),
set $y = u$ (the input) and $\Pi = I_N$, and adopt the combined prior
$\mcR(x) = \lambda\|x\|_1 + \iota_{\{x\geq 0\}}$ (Laplace plus
nonnegativity).  The MAP objective becomes the positive LASSO problem
$\min_{x\geq 0}\tfrac{1}{2}\|u-\Phi x\|^2+\lambda\|x\|_1$, and
$\prox_{\mcR}(v) = \relu(v - \lambda\,\vectorones)$.
Proposition~\ref{prop:firing-rate} gives
\begin{equation}
  \tau\dot{x} \;=\; -x \;+\; \relu\bigl((I_M-\Phi^\top\Phi)\,x \;+\; \Phi^\top u \;-\; \lambda\,\vectorones\bigr), \label{eq:pcn}
\end{equation}
which is the \emph{positive competitive network} of
\citep{VC-AG-AD-GR-FB:23a}: a recurrent firing-rate dynamics with
nonnegative states, lateral inhibition, and a global silence threshold that
solves the positive LASSO.  Equation~\eqref{eq:pcn} is related to, but
distinct from, the locally-competitive algorithm (LCA) of
\citep{CJR-DHJ-RGB-BAO:08}.  In particular, the LCA evolves an internal
voltage variable and applies a soft-thresholding readout, whereas
\eqref{eq:pcn} evolves the firing rate directly and uses the proximal
operator of the Laplace-plus-nonnegativity prior.  In either form, the
effective synaptic matrix $W = I_M-\Phi^\top\Phi$ encodes lateral
inhibition.  Indeed, the off-diagonal entry $-\Phi_i^\top\Phi_j$ suppresses
neuron $j$ whenever its basis vector $\Phi_j$ overlaps with $\Phi_i$, which
automatically enforces a sparse code, and the global threshold $\lambda$
silences the neurons whose feedforward drive $[\Phi^\top u]_i$ falls below
it.

\paragraph{Categorical inference and the softmax.}
Consider the \emph{shifted entropy barrier} prior
$\mcR(x) = \entropybar(x) := -\theta\,\entropy(x) -
\tfrac{1}{2}\|x\|^2$ for $x\in\simplex{n}$, and $+\infty$ otherwise, with
temperature $\theta \geq 1$ (Table~\ref{tab:priors}).  By
\citep[Exercise~2.23]{PLC-JCP:20b} and \citep{FR-VC-FB-GR:25k}, the function
$\entropybar$ is closed, convex, and proper, and its proximal
operator is exactly the softmax $\softmax{\theta}$ of
Subsection~\ref{sec:prelim-priors}:
\begin{equation}
  \prox_{\entropybar}(v) \;=\; \softmax{\theta}(v) .
  \label{eq:softmax-prox}
\end{equation}
With prior $\mcR = \entropybar$,
Proposition~\ref{prop:firing-rate} yields
\begin{equation}
  \tau\,\dot{x} \;=\; -x \;+\; \softmax{\theta}\bigl(Wx + \Phi^\top\Pi y\bigr), \label{eq:softmax-dyn}
\end{equation}
with $W = I_M - \Phi^\top\Pi\Phi$.  At equilibrium,
$x^*_i \propto \exp(u^*_i/\theta)$, where $u^* = Wx^* + \Phi^\top\Pi y$ is
the net drive; that is, the steady state is a \emph{Gibbs distribution} over
the categories, and the temperature $\theta$ controls its sharpness.  A
small $\theta$ sharpens the output toward the argmax (categorical
perception), whereas a large $\theta$ broadens it toward the uniform
distribution (maximal uncertainty).

\paragraph{Regularized energy as an extended Hopfield energy.}
Let $y$ be a vector of membrane potentials and let $\sigma$ be a smooth and
monotonically increasing scalar activation applied component-wise.  In this
paragraph only, $W=W^\top$ denotes the classical Hopfield recurrent weight
matrix (distinct from the predictive-coding dictionary $\Phi$ used
elsewhere).  The classical Hopfield model and its Lyapunov energy are
\begin{equation}
  \dot{y} \;=\; -y \;+\; W\sigma(y) ,
  \qquad
  \subscr{\mcE}{Hopfield}(y)
  \;=\; -\tfrac{1}{2}\,\sigma(y)^\top W \sigma(y)
  \;+\; \sum_{i=1}^{n} \int_0^{\sigma_i(y_i)} \sigma_i^{-1}(z)\,dz .
  \label{eq:hopfield-y}
\end{equation}
After the change of variables to the firing rate $x = \sigma(y)$, the
energy reads~\citep{SB-GB-FB-SZ:24m}
\begin{equation}
  \subscr{\mcE}{Hopfield}(x)
  \;=\; -\tfrac{1}{2}\,x^\top W x
  \;+\; \sum_{i=1}^{n} \int_0^{x_i} \sigma_i^{-1}(z)\,dz
  \;+\; \iota_{\operatorname{image}(\sigma)}(x) ,
  \label{eq:hopfield}
\end{equation}
where the indicator $\iota_{\operatorname{image}(\sigma)}$ enforces that
$x$ stays in the range of $\sigma$.  Loosely speaking (other splittings are
possible), energy~\eqref{eq:hopfield} has the generic regularized form
$\mcE + \mcR$ of Subsection~\ref{sec:prelim-pgd}: the smooth part is the
quadratic $\mcE(x) = -\tfrac{1}{2}x^\top W x$, and the regularizer
\begin{equation}
  \mcR(x) \;=\; \sum_{i=1}^{n} \int_0^{x_i} \sigma_i^{-1}(z)\,dz
  \;+\; \iota_{\operatorname{image}(\sigma)}(x)
\end{equation}
has the activation $\sigma$ as its proximal operator, by the convex
conjugate (Legendre transform) pairing between $\mcR$ and $\sigma$.
Proximal gradient descent on the regularized energy~\eqref{eq:hopfield}
therefore gives rise to a firing-rate neural network that corresponds to the
membrane-potential Hopfield model~\eqref{eq:hopfield-y}.  The same
construction extends the classical Hopfield energy to nonsmooth and
constrained activations: each row of Table~\ref{tab:priors} substitutes a
different $\sigma$ (soft-threshold for the Laplace prior, ReLU for
nonnegativity, softmax for the shifted-entropy-barrier) by the same
Legendre-pair mechanism.

%% \clearpage

\section{Hierarchical Predictive Coding as a Network of Leaky Firing-Rate Networks}
\label{sec:hier}

The cortex is hierarchical, and predictive coding extends to multiple levels
by stacking single-level circuits \citep{RPNR-DHB:99, KJF:08,
AMB-WMU-RAA-GRM-PF-KJF:12}.  We present the hierarchy in four steps: the
centralized MAP problem for a deep linear generative model
(§\ref{sec:centralized}); its variable-splitting reformulation as
hierarchical predictive coding (§\ref{sec:varsplit}); a probabilistic
graphical model perspective, which shows that the same energy arises from
conditional independence (§\ref{sec:pgm}); and the distributed proximal
gradient descent dynamics that result (§\ref{sec:hierdyn}).

\subsection{The Centralized MAP Problem}\label{sec:centralized}
Consider a generative model with a latent variable
$z \in \real^{M_L}$ and a composite deep linear forward map,
$$ y \;=\;
\Phi^{(1)} \Phi^{(2)} \cdots \Phi^{(L)}\, z \;+\; \eta, \qquad \eta \sim \mathcal{N}\bigl(0,\; \Pi^{-1}\bigr), $$
and prior $p(z) \propto \exp(-\mcR(z))$.  MAP inference solves
\begin{equation}
    \min_{z} \;\; \tfrac{1}{2}\bigl\| y - \Phi^{(1)} \Phi^{(2)} \cdots \Phi^{(L)}\, z \bigr\|_{\Pi}^2 \;+\; \mcR(z) .
\label{eq:centralized}
\end{equation}
Three features of \eqref{eq:centralized} are problematic for a brain-like
solver.  First, the gradient requires the transpose of the full product,
$\Phi^{(L)\top}\cdots \Phi^{(1)\top}$, which is nonlocal.  Second, the deep
products $\Phi^{(1)}\cdots \Phi^{(L)}$ are typically poorly conditioned.
Third, only the top-level latent $z$ is regularized, which is inconsistent
with the various statistical properties observed at every level of the
visual cortex.  Variable splitting removes all three drawbacks.

\subsection{Variable Splitting: From Flat to Hierarchical}\label{sec:varsplit}
We introduce one auxiliary variable $x^{(\ell)} \in \real^{M_\ell}$ per
level, with $x^{(L)} \equiv z$ and $x^{(0)} \equiv y$.  The constrained
problem
\begin{equation}
    \min_{x^{(1:L)}} \;\; \tfrac{1}{2}\bigl\|y - \Phi^{(1)} x^{(1)}\bigr\|_\Pi^2 + \mcR\bigl(x^{(L)}\bigr) \quad \text{s.t.}\quad x^{(\ell-1)} = \Phi^{(\ell)} x^{(\ell)}, \;\;
\ell = 2, \ldots, L,
    \label{eq:constrained}
\end{equation}
is equivalent to \eqref{eq:centralized}, because the elimination of the
constraints by forward substitution reconstructs the deep product.  We then
relax each exact constraint into a Gaussian penalty with level-wise
precision $\Pi^{(\ell)}$ and add a level-wise prior $\mcR^{(\ell)}$ at each
level to obtain:
\begin{equation}
    \min_{x^{(1:L)}} \;\;
    \subscr{\mcE}{H-MAP}\bigl(x^{(1:L)}\bigr) \;:=\;
\underbrace{\sum_{\ell=1}^{L} \tfrac{1}{2}\,\bigl\|x^{(\ell-1)} - \Phi^{(\ell)} x^{(\ell)}\bigr\|_{\Pi^{(\ell)}}^2}_{\subscr{\mcE}{recon}(x^{(1:L)})} \;+\; \sum_{\ell=1}^{L} \mcR^{(\ell)}\bigl(x^{(\ell)}\bigr).
    \label{eq:hier-map}
\end{equation}
We refer to $\subscr{\mcE}{H-MAP}$ as the \emph{hierarchical MAP energy}, the
defining energy of hierarchical predictive coding \citep{RPNR-DHB:99, KJF:08}.  For finite precisions
$\Pi^{(\ell)}$, equation \eqref{eq:hier-map} is an approximation of
\eqref{eq:constrained}, and not an equivalence; only in the limit
$\Pi^{(\ell)} \to \infty$ are the hard constraints of
\eqref{eq:constrained} recovered, and, with $\mcR^{(\ell)} \equiv 0$ for
$\ell < L$, we further recover \eqref{eq:centralized}.  The transformation
from \eqref{eq:centralized} to \eqref{eq:hier-map} is therefore best
understood as a classical variable-splitting relaxation of the Alternating
Direction Method of Multipliers (ADMM) type \citep{SP-NP-EC-NP-JE:10,
  NP-SB:14, PLC-JCP:11}.  The same auxiliary-variable device underlies the
layer-wise training of deep models \citep{MCP-WW:14, GT-AB-MX-MS-AP-TG:16}.

\subsection{A Probabilistic Graphical Model Perspective}\label{sec:pgm}

The hierarchical energy $\subscr{\mcE}{H-MAP}$ defined
in~\eqref{eq:hier-map} admits a probabilistic interpretation that supports
the locality of the proximal-gradient dynamics directly, without the use of
variable splitting, provided that the underlying graphical model is read
with care.  We emphasize at the outset that the variable-splitting
derivation of §\ref{sec:varsplit} and the present probabilistic graphical
model reading correspond to two distinct probabilistic models that share
the same energy function, and not to two derivations of the same model.
The derivation of §\ref{sec:varsplit} treats $\subscr{\mcE}{H-MAP}$ as a
finite-precision relaxation of the directed-chain MAP problem
\eqref{eq:centralized}, whereas the present section takes
$\subscr{\mcE}{H-MAP}$ to be the exact joint energy of an undirected Markov
random field.

A directed generative chain $x^{(L)} \to \cdots \to x^{(1)} \to y$ with
linear-Gaussian conditionals $p(x^{(\ell)}\mid x^{(\ell+1)}) =
\mathcal{N}(\Phi^{(\ell+1)}x^{(\ell+1)},(\Pi^{(\ell+1)})^{-1})$ and
top-level prior $p(x^{(L)})$ fixes the joint distribution $p(y, x^{(1:L)})$
completely.  The marginal of every intermediate $x^{(\ell)}$, with $\ell <
L$, is then \emph{induced} by the chain, and it cannot be replaced by a
separate ``independent'' prior $\exp(-\mcR^{(\ell)}(x^{(\ell)}))$ without
changing the model.  The augmentation of the directed chain by intermediate
node factors $\mcR^{(\ell)}$, $\ell < L$, therefore requires a different
probabilistic interpretation.

The mathematically consistent interpretation is as an \emph{undirected
Markov random field} (MRF) (equivalently, a chain-structured factor graph
or energy-based model),
\begin{equation}
  p(y, x^{(1:L)}) \;=\; \frac{1}{Z}\,\exp\bigl(-\subscr{\mcE}{H-MAP}(x^{(1:L)})\bigr) ,
  \label{eq:mrf-joint}
\end{equation}
with normalizing constant $Z$ and the energy of \eqref{eq:hier-map}
read as
\begin{equation}
  \subscr{\mcE}{H-MAP}(x^{(1:L)}) \;=\;
  \sum_{\ell=1}^{L} \underbrace{\tfrac{1}{2}\bigl\|x^{(\ell-1)} - \Phi^{(\ell)}x^{(\ell)}\bigr\|_{\Pi^{(\ell)}}^2}_{\text{pairwise potential on edge }(\ell-1,\ell)} \;+\; \sum_{\ell=1}^{L} \underbrace{\mcR^{(\ell)}(x^{(\ell)})}_{\text{node potential at vertex }\ell} ,
\end{equation}
with $x^{(0)} \equiv y$ clamped.  Each pairwise potential encodes the
linear-Gaussian compatibility between adjacent levels, and each node
potential $\mcR^{(\ell)}$ encodes the level-wise constraint or regularizer
whose proximal operator delivers the activation function at level $\ell$
(Table~\ref{tab:priors}).  The MRF is well-defined whenever $Z<\infty$,
which holds in particular when every $\mcR^{(\ell)}$ is proper, closed, and
convex with bounded effective domain or with at-least-linear coercive
growth.

The chain MRF has the Markov property: each $x^{(\ell)}$ is conditionally
independent of all other levels given its two neighbors.  Therefore, the
level-$\ell$ gradient $\nabla_{x^{(\ell)}}\subscr{\mcE}{recon}$ of the
reconstruction energy introduced in \eqref{eq:hier-map} depends only on
$\epsilon^{(\ell)}$ and $\epsilon^{(\ell+1)}$.  The locality of the
proximal-gradient drive $u^{(\ell)}$ in Proposition~\ref{prop:hier-pgd} is
therefore a probabilistic consequence of the MRF structure.  The directed
chain is recovered as the special case $\mcR^{(\ell)} \equiv 0$ for $\ell < L$
and $\mcR^{(L)} = -\log p(x^{(L)})$; for general intermediate
$\mcR^{(\ell)}$, however, no directed factorization exists.  Since $Z$ does
not depend on $x^{(1:L)}$, MAP inference on the MRF coincides with energy
minimization, and the hierarchical proximal-gradient dynamics
\eqref{eq:hier-pgd} solve the same optimization problem in both readings.

\begin{remark}[Comparison]
  Three features distinguish the MRF interpretation from the variable-splitting
  derivation of §\ref{sec:varsplit}.  First, locality in the MRF is grounded in the
  conditional independence of an explicit joint distribution, rather than in
  algebraic decoupling.  Second, the level-wise $\mcR^{(\ell)}$ for
  $\ell < L$ are not Bayesian priors on independent variables, which would
  over-specify the directed chain, but rather node potentials of the chain
  MRF~\eqref{eq:mrf-joint}.  Third, $\prox_{\mcR^{(\ell)}}$ is
  active at every level, rather than only at the top level $\ell=L$, and
  this allows different areas to implement different activation functions
  (Table~\ref{tab:priors}).
\end{remark}

\subsection{Hierarchical Proximal Gradient Descent Dynamics}\label{sec:hierdyn}

We apply continuous-time proximal gradient descent level-wise to the
hierarchical MAP energy \eqref{eq:hier-map}, with the sensory data clamped
as $x^{(0)} \equiv y$.  This yields the following distributed system, which
is illustrated in Figure~\ref{fig:pgd-arch}.
\begin{proposition}[Hierarchical Proximal Gradient Descent]\label{prop:hier-pgd}
Fix the following data:
\begin{itemize}
\item synaptic weight matrices $\Phi^{(\ell)} \in \real^{M_{\ell-1} \times M_\ell}$ for $\ell \in \until{L}$,
  for an integer $L \geq 1$ and dimensions $M_0, M_1, \ldots, M_L \geq 1$;
\item likelihood precisions $\Pi^{(\ell)} \in \real^{M_{\ell-1} \times M_{\ell-1}}$, symmetric positive-definite, for $\ell \in \until{L}$;
\item regularizers / negative log-priors $\mcR^{(\ell)} : \real^{M_\ell} \to \real \cup \{+\infty\}$, proper, closed, and convex, for $\ell \in \until{L}$;
\item time constants $\tau^{(\ell)} > 0$ for $\ell \in \until{L}$;
\item a fixed observation $y \in \real^{M_0}$, with the convention $x^{(0)} \equiv y$.
\end{itemize}
Define the \emph{hierarchical MAP energy}
$$ \subscr{\mcE}{H-MAP}\bigl(x^{(1:L)}\bigr) \;:=\;
\sum_{\ell=1}^{L} \tfrac{1}{2}\bigl\|x^{(\ell-1)} - \Phi^{(\ell)} x^{(\ell)}\bigr\|_{\Pi^{(\ell)}}^2 \;+\; \sum_{\ell=1}^{L} \mcR^{(\ell)}\bigl(x^{(\ell)}\bigr) , $$
the \emph{level-wise prediction errors}
$\epsilon^{(\ell)} := x^{(\ell-1)} - \Phi^{(\ell)} x^{(\ell)} \in \real^{M_{\ell-1}}$,
and the \emph{net precision-weighted error currents}
$$ u^{(\ell)} \;:=\;
\Phi^{(\ell)\top}\Pi^{(\ell)}\epsilon^{(\ell)} \;-\; \Pi^{(\ell+1)}\epsilon^{(\ell+1)} \quad (\ell < L), \qquad u^{(L)} \;:=\; \Phi^{(L)\top}\Pi^{(L)}\epsilon^{(L)} .
$$
Then continuous-time proximal-gradient descent on
$\subscr{\mcE}{H-MAP}$ is the level-wise distributed system
\begin{equation}
\tau^{(\ell)}\,\dot{x}^{(\ell)} \;=\;
-x^{(\ell)} + \prox_{\mcR^{(\ell)}}\bigl(x^{(\ell)} + u^{(\ell)}\bigr), \qquad \ell \in \until{L} .
\label{eq:hier-pgd}
\end{equation}
Each level-$\ell$ update is an instance of \eqref{eq:pgd-network}
driven by $u^{(\ell)}$.  Equivalently, stacking
$\mathbf{x} = (x^{(1)}, \ldots, x^{(L)})$ and defining the
block-lower-bidiagonal matrix
$$
\boldsymbol{\Phi} \;:=\;
\begin{pmatrix}
\Phi^{(1)} & & & \\
-I_{M_1} & \Phi^{(2)} & & \\
& -I_{M_2} & \ddots & \\
& & \ddots & \Phi^{(L)}
\end{pmatrix}
\;\in\; \real^{(M_0+\cdots+M_{L-1})\times(M_1+\cdots+M_L)},
$$
where the diagonal block $(\ell,\ell)$ is
$\Phi^{(\ell)}\in\real^{M_{\ell-1}\times M_\ell}$ and the
sub-diagonal block $(\ell{+}1,\ell)$ is
$-I_{M_\ell}\in\real^{M_\ell\times M_\ell}$, the stacked error
vector satisfies
$$
\boldsymbol{\epsilon} \;:=\; \bigl(\epsilon^{(1)},\ldots,\epsilon^{(L)}\bigr)
\;=\; \mathbf{b} - \boldsymbol{\Phi}\mathbf{x}, \qquad
\mathbf{b} \;:=\; (y,\,0,\,\ldots,\,0)^\top ,
$$
and with
$\boldsymbol{\Pi} = \operatorname{blockdiag}(\Pi^{(\ell)})$ and
$\prox_{\boldsymbol{\mcR}}$ block-separable,
the dynamics takes the single block-vector form
\begin{equation}
\boldsymbol{\tau}\,\dot{\mathbf{x}} \;=\;
-\,\mathbf{x} \;+\; \prox_{\boldsymbol{\mcR}}\bigl(
(I_{M_1+\cdots+M_L} - \boldsymbol{\Phi}^\top\boldsymbol{\Pi}\boldsymbol{\Phi})\,\mathbf{x}
\;+\; \boldsymbol{\Phi}^\top\boldsymbol{\Pi}\,\mathbf{b}
\bigr). \label{eq:hier-pgd-x}
\end{equation}
\emph{Hierarchical predictive coding is therefore itself a single
proximal-gradient system on the joint energy $\subscr{\mcE}{H-MAP}$.}
\end{proposition}

\begin{proof}
\textit{Step 1: Identify the smooth part of the hierarchical energy.}
The joint energy \eqref{eq:hier-map} splits as
$\subscr{\mcE}{H-MAP} = \subscr{\mcE}{recon} + \sum_\ell \mcR^{(\ell)}$,
where the smooth (differentiable) part is
$$ \subscr{\mcE}{recon} \;=\; \sum_{\ell=1}^{L} \mcE^{(\ell)}, \qquad \mcE^{(\ell)}\bigl(x^{(\ell-1)}, x^{(\ell)}\bigr) \;:=\;
\tfrac{1}{2}\bigl\|x^{(\ell-1)} - \Phi^{(\ell)} x^{(\ell)}\bigr\|_{\Pi^{(\ell)}}^2 $$
and the nonsmooth part $\sum_\ell \mcR^{(\ell)}$ is separable across
levels.  Define the level-$\ell$ prediction error
$\epsilon^{(\ell)} := x^{(\ell-1)} - \Phi^{(\ell)} x^{(\ell)}$.

\textit{Step 2: Compute the gradient with respect to each level.}
For an intermediate level $1 \leq \ell \leq L-1$, the variable
$x^{(\ell)}$ participates in exactly two energy terms:
\begin{itemize}
    \item $\mcE^{(\ell)}$, where $x^{(\ell)}$ appears as the \emph{generator} of the prediction at level $\ell{-}1$: $\nabla_{x^{(\ell)}} \mcE^{(\ell)} = -\Phi^{(\ell)\top}\Pi^{(\ell)}\epsilon^{(\ell)}$.
    \item $\mcE^{(\ell+1)}$, where $x^{(\ell)}$ appears as the \emph{prediction target} from level $\ell{+}1$: $\nabla_{x^{(\ell)}} \mcE^{(\ell+1)} = \Pi^{(\ell+1)}\epsilon^{(\ell+1)}$.
\end{itemize}
Combining (and negating to get the descent direction):
$$-\nabla_{x^{(\ell)}} \subscr{\mcE}{recon} \;=\; \underbrace{\Phi^{(\ell)\top}\Pi^{(\ell)}\epsilon^{(\ell)}}_{\text{bottom-up drive}} \;-\; \underbrace{\Pi^{(\ell+1)}\epsilon^{(\ell+1)}}_{\text{top-down correction}} \;=:\; u^{(\ell)}.$$

\textit{Step 3: Boundary conditions.}
\emph{Top level $\ell = L$:} There is no term $\mcE^{(L+1)}$, so the
top-down correction is absent:
$-\nabla_{x^{(L)}} \subscr{\mcE}{recon} = \Phi^{(L)\top}\Pi^{(L)}\epsilon^{(L)} =: u^{(L)}.$
\emph{Bottom level:} $x^{(0)} \equiv y$ is clamped and enters only as
a constant in $\epsilon^{(1)} = y - \Phi^{(1)} x^{(1)}$.

\textit{Step 4: Apply the proximal gradient update level-wise.}
The nonsmooth term $\sum_\ell \mcR^{(\ell)}(x^{(\ell)})$ is separable
across levels, so its proximal operator is block-separable:
$$ \prox_{\boldsymbol{\mcR}}\bigl(\mathbf{x}\bigr) \;=\; \bigl(\prox_{\mcR^{(\ell)}}(x^{(\ell)})\bigr)_{\ell=1}^L. $$
The generic continuous-time proximal gradient descent rule
$\tau\dot{x} = -x + \prox_{\mcR}(x - \nabla\mcE(x))$
applied block-level-wise with time constant $\tau^{(\ell)}$ yields
\eqref{eq:hier-pgd}.

\textit{Step 5: Block-vector form.}
Stacking $\mathbf{x}$ and
$\boldsymbol{\epsilon} = (\epsilon^{(1)},\ldots,\epsilon^{(L)})$,
inspection of the block rows of $\boldsymbol{\Phi}^\top$ gives
$-\nabla_\mathbf{x}\subscr{\mcE}{recon} = \boldsymbol{\Phi}^\top\boldsymbol{\Pi}\boldsymbol{\epsilon}$.
Since $\boldsymbol{\epsilon} = \mathbf{b} - \boldsymbol{\Phi}\mathbf{x}$, the
smooth gradient equals
$\boldsymbol{\Phi}^\top\boldsymbol{\Pi}(\mathbf{b}-\boldsymbol{\Phi}\mathbf{x})$, and
substituting into the level-wise update yields \eqref{eq:hier-pgd-x}.
\end{proof}

%% \begin{remark}[Convergence of the hierarchical dynamics]
%%   As in the single-level case, convergence properties of the hierarchical
%%   proximal-gradient dynamics \eqref{eq:hier-pgd}--\eqref{eq:hier-pgd-x} can
%%   be analyzed through several approaches, with the block Gram matrix
%%   $\boldsymbol{\Phi}^\top\boldsymbol{\Pi}\boldsymbol{\Phi}$ playing the role of the
%%   single-level Gram; see Remark~\ref{remark:convergence}.
%% \end{remark}

\tikzset{
  pgdfig/box/.style={draw, rounded corners=2pt, align=center, inner sep=4pt},
  pgdfig/pe/.style={pgdfig/box, fill=gnblue4!25},
  pgdfig/sum/.style={draw, circle, minimum size=6mm, inner sep=0pt}
}

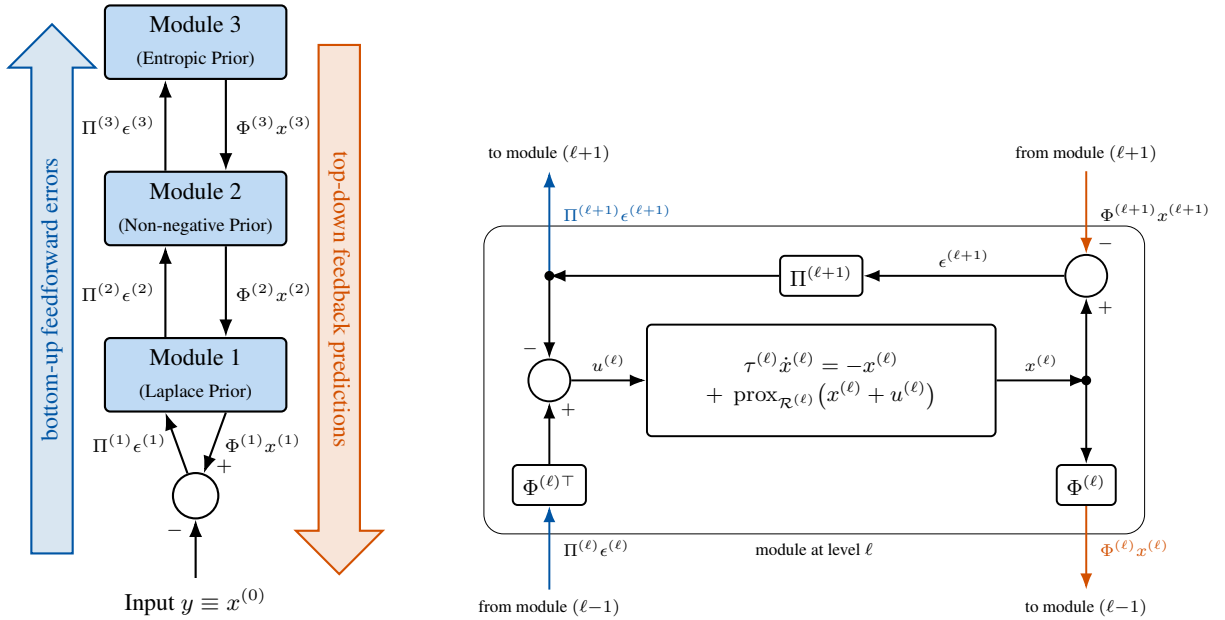
\begin{figure}[htbp]
\centering
%% =============== Panel (a): vertical hierarchy ===============
\begin{subfigure}[t]{0.36\textwidth}
 \centering
\begin{tikzpicture}[font=\small, >=Latex, thick,
    box/.style={pgdfig/box}, pe/.style={pgdfig/pe}, sum/.style={pgdfig/sum}]
\node (input) at (0, -3.0) {Input $y \equiv x^{(0)}$};
\node[sum] (s1) at (0, -1.6) {};
\node[pe, minimum width=2.4cm] (pe1) at (0,  0.0)
     {Module 1\\[2pt]\scriptsize(Laplace Prior)};
\node[pe, minimum width=2.4cm] (pe2) at (0,  2.2)
    {Module 2\\[2pt]\scriptsize(Non-negative Prior)};
\node[pe, minimum width=2.4cm] (pe3) at (0,  4.4)
     {Module 3\\[2pt]\scriptsize(Entropic Prior)};
\draw[->] (input) -- node[left=1pt, font=\scriptsize, pos=0.82]{$-$} (s1.south);
\draw[->] (s1.110)
    -- node[left, font=\scriptsize]{$\Pi^{(1)}\epsilon^{(1)}$}
       ([xshift=-4mm]pe1.south);
\foreach \below/\above/\i in {pe1/pe2/2, pe2/pe3/3} {
    \draw[->] ([xshift=-4mm]\below.north)
        -- node[left, font=\scriptsize]{$\Pi^{(\i)}\epsilon^{(\i)}$}
           ([xshift=-4mm]\above.south);
}
\foreach \above/\below/\i in {pe2/pe1/2, pe3/pe2/3} {
    \draw[->] ([xshift=4mm]\above.south)
        -- node[right, font=\scriptsize]{$\Phi^{(\i)} x^{(\i)}$}
           ([xshift=4mm]\below.north);
}
\draw[->] ([xshift=4mm]pe1.south)
    -- node[right, font=\scriptsize, pos=0.5]{$\Phi^{(1)} x^{(1)}$}
       node[right, font=\scriptsize, pos=0.88]{$+$}
       (s1.70);
\node[single arrow, rotate=90,
      draw=gnblue3, fill=gnblue3!15,
      minimum width=0.9cm, minimum height=7.0cm,
      single arrow head extend=3mm,
      font=\small, text=gnblue3, align=center]
      at (-1.9, 1.0) {bottom-up feedforward errors};
\node[single arrow, rotate=-90,
      draw=gnbrown4, fill=gnbrown4!15,
      minimum width=0.9cm, minimum height=7.0cm,
      single arrow head extend=3mm,
      font=\small, text=gnbrown4, align=center]
      at (1.9, 1.0) {top-down feedback predictions};
\end{tikzpicture}
 \caption{Vertical hierarchy with 3 representative levels. Feedforward
   pathways (left) carry the precision-weighted prediction errors
   $\Pi^{(\ell)}\epsilon^{(\ell)}$ upward; feedback pathways (right)
   carry the predictions $\Phi^{(\ell)}x^{(\ell)}$ downward, following
   \citep[Fig.\ 1]{RPNR-DHB:99}.}
 \label{fig:pgd-arch-a}
 \end{subfigure}
\hfill
%% =============== Panel (b): module at generic level ell ===============
\begin{subfigure}[t]{0.60\textwidth}
\centering
\resizebox{\linewidth}{!}{%
\begin{tikzpicture}[font=\small, >=Latex, thick,
    box/.style={pgdfig/box}, pe/.style={pgdfig/pe}, sum/.style={pgdfig/sum}]

\node[box, minimum width=5.0cm, minimum height=1.6cm] (state) at (4.7, 0)
    {$\tau^{(\ell)}\dot{x}^{(\ell)} = -x^{(\ell)}$\\[2pt]
     $+\;\prox_{\mcR^{(\ell)}}\bigl(x^{(\ell)} + u^{(\ell)}\bigr)$};
\node[sum] (uSum) at (0.8, 0) {};
\node[box] (WT)   at (0.8, -1.5) {$\Phi^{(\ell)\top}$};
\node[sum] (sEps) at (8.5, 1.5) {};
\node[box] (PiUp) at (4.7, 1.5) {$\Pi^{(\ell+1)}$};
\node[box] (Wfb)  at (8.5, -1.5) {$\Phi^{(\ell)}$};

\coordinate (xbranch) at (8.5,  0);
\coordinate (piLeft)  at (0.8,  1.5);

\draw[->, draw=gnblue3] (0.8, -3.0)
    -- node[right=2pt, font=\scriptsize, pos=0.5, align=left, text=black]{$\Pi^{(\ell)}\epsilon^{(\ell)}$} (WT.south);
\draw[->] (WT.north)
    -- node[right=1pt, font=\scriptsize, pos=0.88]{$+$} (uSum.south);
\draw[->] (uSum.east)
    -- node[above, font=\scriptsize]{$u^{(\ell)}$} (state.west);
\draw[->] (state.east)
    -- node[above, font=\scriptsize]{$x^{(\ell)}$} (xbranch);
\draw[->] (xbranch)
    -- node[right=1pt, font=\scriptsize, pos=0.88]{$+$} (sEps.south);
\draw[->] (xbranch) -- (Wfb.north);
\draw[->, draw=gnbrown4] (Wfb.south)
    -- node[right=2pt, font=\scriptsize, pos=0.5, align=left, text=gnbrown4]{$\Phi^{(\ell)}x^{(\ell)}$} (8.5, -3.0);
\draw[->, draw=gnbrown4] (8.5, 3.0)
    -- node[right=2pt, font=\scriptsize, pos=0.5, align=left, text=black]{$\Phi^{(\ell+1)}x^{(\ell+1)}$}
       node[right=1pt, font=\scriptsize, pos=0.88, text=black]{$-$}
       (sEps.north);
\draw[->] (sEps.west)
    -- node[above, font=\scriptsize]{$\epsilon^{(\ell+1)}$} (PiUp.east);
\draw[->] (PiUp.west) -- (piLeft);
\draw[->, draw=gnblue3] (piLeft)
    -- node[right=2pt, font=\scriptsize, pos=0.6, align=left, text=gnblue3]{$\Pi^{(\ell+1)}\epsilon^{(\ell+1)}$} (0.8, 3.0);
\draw[->] (piLeft)
    -- node[left=1pt, font=\scriptsize, pos=0.88]{$-$} (uSum.north);

\fill (xbranch) circle (2pt);
\fill (piLeft)  circle (2pt);

\node[font=\scriptsize, anchor=north] at (0.8, -3.0) {from module $(\ell{-}1)$};
\node[font=\scriptsize, anchor=south] at (0.8,  3.0) {to module $(\ell{+}1)$};
\node[font=\scriptsize, anchor=south] at (8.5,  3.0) {from module $(\ell{+}1)$};
\node[font=\scriptsize, anchor=north] at (8.5, -3.0) {to module $(\ell{-}1)$};

\begin{scope}[on background layer]
\node[draw, thin, rounded corners=8pt, inner sep=4mm,
      fit=(WT)(uSum)(sEps)(PiUp)(state)(Wfb),
      label={[font=\scriptsize]below:module at level $\ell$}] (peBox) {};
\end{scope}
\end{tikzpicture}%
}
\caption{System theoretic module at generic level $\ell$.
  Bottom-left: $\Pi^{(\ell)}\epsilon^{(\ell)}$ projected by $\Phi^{(\ell)\top}$
  feeds the $+$ port of the left junction.
  Top-right: $\Phi^{(\ell+1)}x^{(\ell+1)}$ is subtracted from $x^{(\ell)}$,
  precision-weighted by $\Pi^{(\ell+1)}$, and fed to the $-$ port.
  Net drive $u^{(\ell)}$ enters the leaky-prox state equation.
  Outgoing: $\Pi^{(\ell+1)}\epsilon^{(\ell+1)}$ (top-left) and
  $\Phi^{(\ell)}x^{(\ell)}$ (bottom-right).}
\label{fig:pgd-arch-b}
\end{subfigure}
\caption{Hierarchical predictive coding recast as proximal gradient dynamics.}
\label{fig:pgd-arch}
\end{figure}

%% \paragraph{Per-level time constants and the cortical timescale hierarchy.}
%% The per-level time constants $\tau^{(\ell)}$ in \eqref{eq:hier-pgd} acquire
%% a concrete biological meaning from the empirically established gradient of
%% intrinsic timescales along the cortical hierarchy --- from $\sim100$~ms
%% in sensory cortex to $\sim500$~ms in prefrontal cortex
%% \cite{JDM-AB-DJF-RR-JDW-XC-CPS-TP-HS-DL-XJW:14, SJK-JD-KJF:08,
%%   RC-AK-MMG-HK-XJW:15} --- so a monotone-increasing $\tau^{(\ell)}$ is the
%% natural algorithmic correlate of slower temporal integration at higher
%% cortical levels.

\paragraph{Example: a hierarchical classification circuit.}
Because hierarchical proximal gradient descent admits a separate regularizer
$\mcR^{(\ell)}$ at each level, which is a node potential of the Markov random
field of §\ref{sec:pgm}, distinct computational modules can be stacked.  A
three-level hierarchy for categorical perception might use the following
modules.
\begin{itemize}
    \item \textbf{Level 1 (early sensory level):} a Laplace prior
      $\mcR^{(1)}$, which yields soft-thresholding dynamics that act as a
      sparse coder for low-level features.
    \item \textbf{Level 2 (intermediate level):} a nonnegativity prior
      $\mcR^{(2)}$, which yields ReLU dynamics for strictly nonnegative
      intermediate features.
    \item \textbf{Level 3 (top level):} an entropic prior $\mcR^{(3)}$ on
      the probability simplex, which yields softmax dynamics that produce a
      mutually exclusive categorical decision.
\end{itemize}
A single hierarchical proximal gradient descent thus recovers the three
level-wise activation functions (soft-threshold, ReLU, softmax) as one
mechanism, and each activation function arises as the proximal operator of
the chosen level prior.  Variable splitting is not only a device to handle
deep linear products; it is also the mechanism that lets a single network
pass from continuous feature extraction to discrete classification.

%% \clearpage
\section{Summary}\label{sec:conclusion}

This paper has proposed flat and hierarchical regularized MAP objectives
for predictive coding and Bayesian inference by means of proximal gradient
descent.  The leak, the effective synaptic matrix $W = I_M - \Phi^\top \Pi
\Phi$, the local drive, and the static nonlinearity of a leaky firing-rate
neuron all emerge from proximal gradient descent on that objective
(Proposition~\ref{prop:firing-rate}).  The choice of the prior fixes the
activation function (Table~\ref{tab:priors}).

Variable splitting extends the single-level construction to a hierarchy and
produces the Rao--Ballard prediction-error circuit as a distributed
proximal gradient descent solver. At each level of the hierarchy, the
regularizers $\mcR^{(\ell)}$ deliver level-specific activation functions
(Proposition~\ref{prop:hier-pgd}).  In particular, a sensory-to-categorical
pipeline (sparse coder $\to$ ReLU features $\to$ softmax decision) is
proposed, without a change of inference method between the continuous and
the discrete levels. This approach is in contrast with the mixed-model
approach of \citep{KJF-TF-FR-PS-GP:17,KJF-TP-BdV:17}.  In this work, the
per-level priors are not a modeling convenience: they encode the assumed
beliefs, metabolic costs, and physical constraints -- and they shape of the
firing-rate nonlinearities as well as the equilibrium toward which the
dynamics converges.

\subsection*{Acknowledgement}
This material is based upon work supported by the Army Research Laboratory
under grant number W911NF-24-1-0228. We thank Professors Cristina Savin
(NYU) and Sara Solla (Northwestern University) for insightful conversations
on probabilistic models.

%% \bibliographystyle{abbrvnat+doi}
%% \bibliography{alias,FB,Main}

\begin{thebibliography}{35}
\providecommand{\natexlab}[1]{#1}
\providecommand{\url}[1]{\texttt{#1}}
\expandafter\ifx\csname href\endcsname\relax
  \providecommand{\doi}[1]{doi: #1}\else
  \providecommand{\doi}[1]{doi: \href{https://doi.org/#1}{\nolinkurl{#1}}}\fi

\bibitem[Attinger et~al.(2017)Attinger, Wang, and Keller]{AA-BW-GBK:17}
A.~Attinger, B.~Wang, and G.~B. Keller.
\newblock Visuomotor coupling shapes the functional development of mouse visual
  cortex.
\newblock \emph{Cell}, 169\penalty0 (7):\penalty0 1291--1302, 2017.
\newblock \doi{10.1016/j.cell.2017.05.023}.

\bibitem[Barrett and Miller(2026)]{LFB-EKM:26}
L.~F. Barrett and E.~K. Miller.
\newblock Categorization is ‘baked’ into the brain.
\newblock \emph{Nature Reviews Neuroscience}, 27\penalty0 (6):\penalty0
  435–456, 2026.
\newblock \doi{10.1038/s41583-026-01036-2}.

\bibitem[Bastos et~al.(2012)Bastos, Usrey, Adams, Mangun, Fries, and
  Friston]{AMB-WMU-RAA-GRM-PF-KJF:12}
A.~M. Bastos, W.~M. Usrey, R.~A. Adams, G.~R. Mangun, P.~Fries, and K.~J.
  Friston.
\newblock Canonical microcircuits for predictive coding.
\newblock \emph{Neuron}, 76\penalty0 (4):\penalty0 695--711, 2012.
\newblock \doi{10.1016/j.neuron.2012.10.038}.

\bibitem[Betteti et~al.(2025)Betteti, Baggio, Bullo, and
  Zampieri]{SB-GB-FB-SZ:24m}
S.~Betteti, G.~Baggio, F.~Bullo, and S.~Zampieri.
\newblock Firing rate models as associative memory: {Synaptic} design for
  robust retrieval.
\newblock \emph{Neural Computation}, 37\penalty0 (10):\penalty0 1807--1838,
  2025.
\newblock \doi{10.1162/neco.a.28}.

\bibitem[Boyd et~al.(2010)Boyd, Parikh, Chu, Peleato, and
  Eckstein]{SP-NP-EC-NP-JE:10}
S.~Boyd, N.~Parikh, E.~Chu, B.~Peleato, and J.~Eckstein.
\newblock Distributed optimization and statistical learning via the alternating
  direction method of multipliers.
\newblock \emph{Foundations and Trends in Machine Learning}, 3\penalty0
  (1):\penalty0 1--124, 2010.
\newblock \doi{10.1561/2200000016}.

\bibitem[Buckley et~al.(2017)Buckley, Kim, McGregor, and
  Seth]{CLB-CSK-SM-AKS:17}
C.~L. Buckley, C.~S. Kim, S.~McGregor, and A.~K. Seth.
\newblock The free energy principle for action and perception: {A} mathematical
  review.
\newblock \emph{Journal of Mathematical Psychology}, 81:\penalty0 55--79, 2017.
\newblock \doi{10.1016/j.jmp.2017.09.004}.

\bibitem[Carandini and Heeger(2012)]{MC-DJH:12}
M.~Carandini and D.~J. Heeger.
\newblock Normalization as a canonical neural computation.
\newblock \emph{Nature Reviews Neuroscience}, 13\penalty0 (1):\penalty0 51--62,
  2012.
\newblock \doi{10.1038/nrn3136}.

\bibitem[Carreira-Perpi{\~n}{\'a}n and Wang(2014)]{MCP-WW:14}
M.~{\'A}. Carreira-Perpi{\~n}{\'a}n and W.~Wang.
\newblock Distributed optimization of deeply nested systems.
\newblock In \emph{Int.\ Conf.\ Artificial Intelligence and Statistics
  (AISTATS)}, Proceedings of Machine Learning Research, pages 10--19,
  Reykjavik, Iceland, 2014. PMLR.
\newblock URL \url{https://proceedings.mlr.press/v33/carreira-perpinan14.html}.

\bibitem[Centorrino et~al.(2023)Centorrino, Gokhale, Davydov, Russo, and
  Bullo]{VC-AG-AD-GR-FB:23c}
V.~Centorrino, A.~Gokhale, A.~Davydov, G.~Russo, and F.~Bullo.
\newblock Euclidean contractivity of neural networks with symmetric weights.
\newblock \emph{IEEE Control Systems Letters}, 7:\penalty0 1724--1729, 2023.
\newblock \doi{10.1109/LCSYS.2023.3278250}.

\bibitem[Centorrino et~al.(2024{\natexlab{a}})Centorrino, Davydov, Gokhale,
  Russo, and Bullo]{VC-AD-AG-GR-FB:24a}
V.~Centorrino, A.~Davydov, A.~Gokhale, G.~Russo, and F.~Bullo.
\newblock On weakly contracting dynamics for convex optimization.
\newblock \emph{IEEE Control Systems Letters}, 8:\penalty0 1745--1750,
  2024{\natexlab{a}}.
\newblock \doi{10.1109/LCSYS.2024.3414348}.

\bibitem[Centorrino et~al.(2024{\natexlab{b}})Centorrino, Gokhale, Davydov,
  Russo, and Bullo]{VC-AG-AD-GR-FB:23a}
V.~Centorrino, A.~Gokhale, A.~Davydov, G.~Russo, and F.~Bullo.
\newblock Positive competitive networks for sparse reconstruction.
\newblock \emph{Neural Computation}, 36\penalty0 (6):\penalty0 1163–1197,
  2024{\natexlab{b}}.
\newblock \doi{10.1162/neco_a_01657}.

\bibitem[Combettes and Pesquet(2011)]{PLC-JCP:11}
P.~L. Combettes and J.-C. Pesquet.
\newblock \emph{Proximal Splitting Methods in Signal Processing}, page
  185–212.
\newblock Springer New York, 2011.
\newblock ISBN 9781441995698.
\newblock \doi{10.1007/978-1-4419-9569-8_10}.

\bibitem[Combettes and Pesquet(2020)]{PLC-JCP:20b}
P.~L. Combettes and J.-C. Pesquet.
\newblock Deep neural network structures solving variational inequalities.
\newblock \emph{Set-Valued and Variational Analysis}, 28\penalty0 (3):\penalty0
  491--518, 2020.
\newblock \doi{10.1007/s11228-019-00526-z}.

\bibitem[Davydov et~al.(2025)Davydov, Centorrino, Gokhale, Russo, and
  Bullo]{AD-VC-AG-GR-FB:23f}
A.~Davydov, V.~Centorrino, A.~Gokhale, G.~Russo, and F.~Bullo.
\newblock Time-varying convex optimization: A contraction and equilibrium
  tracking approach.
\newblock \emph{IEEE Transactions on Automatic Control}, 70\penalty0
  (11):\penalty0 7446--7460, 2025.
\newblock \doi{10.1109/TAC.2025.3576043}.

\bibitem[Friston(2005)]{KJF:05}
K.~J. Friston.
\newblock A theory of cortical responses.
\newblock \emph{Philosophical Transactions of the Royal Society B},
  360\penalty0 (1456):\penalty0 815--836, 2005.
\newblock \doi{10.1098/rstb.2005.1622}.

\bibitem[Friston(2008)]{KJF:08}
K.~J. Friston.
\newblock Hierarchical models in the brain.
\newblock \emph{PLoS Computational Biology}, 4\penalty0 (11):\penalty0
  e1000211, 2008.
\newblock \doi{10.1371/journal.pcbi.1000211}.

\bibitem[Friston(2010)]{KJF:10}
K.~J. Friston.
\newblock The free-energy principle: {A} unified brain theory?
\newblock \emph{Nature Reviews Neuroscience}, 11\penalty0 (2):\penalty0
  127--138, 2010.
\newblock \doi{10.1038/nrn2787}.

\bibitem[Friston et~al.(2017{\natexlab{a}})Friston, {FitzGerald}, Rigoli,
  Schwartenbeck, and Pezzulo]{KJF-TF-FR-PS-GP:17}
K.~J. Friston, T.~{FitzGerald}, F.~Rigoli, P.~Schwartenbeck, and G.~Pezzulo.
\newblock Active inference: {A} process theory.
\newblock \emph{Neural Computation}, 29\penalty0 (1):\penalty0 1--49,
  2017{\natexlab{a}}.
\newblock \doi{10.1162/NECO_a_00912}.

\bibitem[Friston et~al.(2017{\natexlab{b}})Friston, Parr, and
  {de~Vries}]{KJF-TP-BdV:17}
K.~J. Friston, T.~Parr, and B.~{de~Vries}.
\newblock The graphical brain: {Belief} propagation and active inference.
\newblock \emph{Network Neuroscience}, 1\penalty0 (4):\penalty0 381--414,
  2017{\natexlab{b}}.
\newblock \doi{10.1162/NETN_a_00018}.

\bibitem[Gokhale et~al.(2024)Gokhale, Davydov, and Bullo]{AG-AD-FB:24d}
A.~Gokhale, A.~Davydov, and F.~Bullo.
\newblock Proximal gradient dynamics: {Monotonicity}, exponential convergence,
  and applications.
\newblock \emph{IEEE Control Systems Letters}, 8:\penalty0 2853--2858, 2024.
\newblock \doi{10.1109/LCSYS.2024.3516632}.

\bibitem[Hassan-Moghaddam and Jovanovi{\'c}(2021)]{SHM-MRJ:21}
S.~Hassan-Moghaddam and M.~R. Jovanovi{\'c}.
\newblock Proximal gradient flow and {D}ouglas-{R}achford splitting dynamics:
  {G}lobal exponential stability via integral quadratic constraints.
\newblock \emph{Automatica}, 123:\penalty0 109311, 2021.
\newblock \doi{10.1016/j.automatica.2020.109311}.

\bibitem[Keller and Mrsic-Flogel(2018)]{GBK-TDM:18}
G.~B. Keller and T.~D. Mrsic-Flogel.
\newblock Predictive processing: {A} canonical cortical computation.
\newblock \emph{Neuron}, 100\penalty0 (2):\penalty0 424--435, 2018.
\newblock \doi{10.1016/j.neuron.2018.10.003}.

\bibitem[Lee and Mumford(2003)]{TSL-DM:03}
T.~S. Lee and D.~Mumford.
\newblock Hierarchical {Bayesian} inference in the visual cortex.
\newblock \emph{Journal of the Optical Society of America A}, 20\penalty0
  (7):\penalty0 1434--1448, 2003.
\newblock \doi{10.1364/josaa.20.001434}.

\bibitem[Marino(2022)]{JM:22}
J.~Marino.
\newblock Predictive coding, variational autoencoders, and biological
  connections.
\newblock \emph{Neural Computation}, 34\penalty0 (1):\penalty0 1--44, 2022.
\newblock \doi{10.1162/neco_a_01458}.

\bibitem[Millidge et~al.(2021)Millidge, Seth, and Buckley]{BM-AS-CLB:21}
B.~Millidge, A.~Seth, and C.~L. Buckley.
\newblock Predictive coding: {A} theoretical and experimental review.
\newblock \emph{arXiv preprint}, 2021.
\newblock \doi{10.48550/arXiv.2107.12979}.

\bibitem[Millidge et~al.(2022)Millidge, Tschantz, and Buckley]{BM-AT-CLB:22}
B.~Millidge, A.~Tschantz, and C.~L. Buckley.
\newblock Predictive coding approximates backprop along arbitrary computation
  graphs.
\newblock \emph{Neural Computation}, 34\penalty0 (6):\penalty0 1329--1368,
  2022.
\newblock \doi{10.1162/neco_a_01497}.

\bibitem[Olshausen and Field(1996)]{BAO-DJF:96}
B.~A. Olshausen and D.~J. Field.
\newblock Emergence of simple-cell receptive field properties by learning a
  sparse code for natural images.
\newblock \emph{Nature}, 381\penalty0 (6583):\penalty0 607--609, 1996.
\newblock \doi{10.1038/381607a0}.

\bibitem[Parikh and Boyd(2014)]{NP-SB:14}
N.~Parikh and S.~Boyd.
\newblock Proximal algorithms.
\newblock \emph{Foundations and Trends in Optimization}, 1\penalty0
  (3):\penalty0 127--239, 2014.
\newblock \doi{10.1561/2400000003}.

\bibitem[Rao and Ballard(1999)]{RPNR-DHB:99}
R.~P.~N. Rao and D.~H. Ballard.
\newblock Predictive coding in the visual cortex: {A} functional interpretation
  of some extra-classical receptive-field effects.
\newblock \emph{Nature Neuroscience}, 2\penalty0 (1):\penalty0 79--87, 1999.
\newblock \doi{10.1038/4580}.

\bibitem[Rossi et~al.(2025)Rossi, Centorrino, Bullo, and
  Russo]{FR-VC-FB-GR:25k}
F.~Rossi, V.~Centorrino, F.~Bullo, and G.~Russo.
\newblock Neural policy composition from free energy minimization.
\newblock \emph{Technical Report}, 2025.
\newblock \doi{10.48550/arXiv.2512.04745}.
\newblock arXiv:2512.04745.

\bibitem[Rozell et~al.(2008)Rozell, Johnson, Baraniuk, and
  Olshausen]{CJR-DHJ-RGB-BAO:08}
C.~J. Rozell, D.~H. Johnson, R.~G. Baraniuk, and B.~A. Olshausen.
\newblock Sparse coding via thresholding and local competition in neural
  circuits.
\newblock \emph{Neural Computation}, 20\penalty0 (10):\penalty0 2526--2563,
  2008.
\newblock \doi{10.1162/neco.2008.03-07-486}.

\bibitem[Salvatori et~al.(2026)Salvatori, Mali, Buckley, Lukasiewicz, Rao,
  Friston, and Ororbia]{TS-AM-CLB-TL-RPNR-KF-AO:26}
T.~Salvatori, A.~Mali, C.~L. Buckley, T.~Lukasiewicz, R.~P. Rao, K.~Friston,
  and A.~Ororbia.
\newblock A survey on neuro-mimetic deep learning via predictive coding.
\newblock \emph{Neural Networks}, 195:\penalty0 108161, 2026.
\newblock \doi{10.1016/j.neunet.2025.108161}.

\bibitem[Taylor et~al.(2016)Taylor, Burmeister, Xu, Singh, Patel, and
  Goldstein]{GT-AB-MX-MS-AP-TG:16}
G.~Taylor, R.~Burmeister, Z.~Xu, B.~Singh, A.~Patel, and T.~Goldstein.
\newblock Training neural networks without gradients: {A} scalable {ADMM}
  approach.
\newblock In \emph{International Conference on Machine Learning}, pages
  2722--2731. PMLR, 2016.
\newblock URL \url{https://proceedings.mlr.press/v48/taylor16.html}.

\bibitem[von Helmholtz(1867)]{HvH:1867}
H.~von Helmholtz.
\newblock \emph{Handbuch der Physiologischen {Optik}}.
\newblock Voss, Leipzig, 1867.

\bibitem[Whittington and Bogacz(2017)]{JCRW-RB:17}
J.~C.~R. Whittington and R.~Bogacz.
\newblock An approximation of the error backpropagation algorithm in a
  predictive coding network with local {Hebbian} synaptic plasticity.
\newblock \emph{Neural Computation}, 29\penalty0 (5):\penalty0 1229--1262,
  2017.
\newblock \doi{10.1162/neco_a_00949}.

\end{thebibliography}

\end{document}